\documentclass{article}

\usepackage{ccfonts}
\usepackage{amssymb,amsmath,amsthm}
\usepackage{mathrsfs}
\usepackage{epsfig}
\usepackage[usenames,dvipsnames]{color}

\newcommand\ie{{\em i.e.}}

\def\B{\mathscr B}
\def\C{\mathbb C}
\def\d{\mathrm{d}}
\def\D{\mathbb D}

\def\G{\mathcal G}
\def\H{\mathcal H}
\def\K{\mathscr K}
\def\M{\mathcal M}
\def\N{\mathbb N}
\def\O{\mathcal O}

\def\R{\mathbb R}
\def\S{\mathscr S}
\def\T{\mathbb T}
\def\V{\mathcal V}
\def\X{\mathfrak X}
\def\Z{\mathbb Z}

\def\dom{\mathcal D}
\def\Ran{\mathop{\mathsf{Ran}}\nolimits}
\def\e{\mathop{\mathrm{e}}\nolimits}
\def\id{\mathop{\mathrm{id}}\nolimits}
\def\det{\mathop{\mathrm{det}}\nolimits}
\def\im{\mathop{\mathsf{Im}}\nolimits}
\def\re{\mathop{\mathsf{Re}}\nolimits}
\def\linf{\mathsf{L}^{\:\!\!\infty}}
\def\ltwo{\mathsf{L}^{\:\!\!2}}

\def\cl{\mathop{\mathsf{cl}}\nolimits}
\def\slim{\mathop{\hbox{\rm s-}\lim}\nolimits}

\newtheorem{Theorem}{Theorem}[section]
\newtheorem{Remark}[Theorem]{Remark}
\newtheorem{Lemma}[Theorem]{Lemma}

\newtheorem{Corollary}[Theorem]{Corollary}
\newtheorem{Proposition}[Theorem]{Proposition}

\newtheorem{Example}[Theorem]{Example}

\begin{document}

%--------------------------------------------------------------------------------------
% Title
%--------------------------------------------------------------------------------------

\title{Commutator methods for unitary operators}

\author{
C. Fern\'andez$^1$\footnote{Supported by the Fondecyt Grant 1100304 and by the
ECOS/CONICYT Grant C10E01.},
S. Richard$^2$\footnote{On leave from Universit\'e de Lyon; Universit\'e Lyon 1;
CNRS, UMR5208, Institut Camille Jordan, 43 blvd du 11 novembre 1918, F-69622
Villeurbanne-Cedex, France. Supported by the Japan Society for the Promotion of
Science (JSPS) and by ``Grants-in-Aid for scientific Research".}
and R. Tiedra de Aldecoa$^1$\footnote{Supported by the Fondecyt Grant 1090008, by the
Iniciativa Cientifica Milenio ICM P07-027-F ``Mathematical Theory of Quantum and
Classical Magnetic Systems" and by the ECOS/CONICYT Grant C10E01.}
}
\date{\small}
\maketitle \vspace{-1cm}

\begin{quote}
\emph{
\begin{itemize}
\item[$^1$] Facultad de Matem\'aticas, Pontificia Universidad Cat\'olica de Chile,\\
Av. Vicu\~na Mackenna 4860, Santiago, Chile
\item[$^2$] Graduate School of Pure and Applied Sciences,
University of Tsukuba, \\
1-1-1 Tennodai,
Tsukuba, Ibaraki 305-8571, Japan
\item[] \emph{E-mails:} cfernand@mat.puc.cl,
richard@math.univ-lyon1.fr, rtiedra@mat.puc.cl
\end{itemize}
  }
\end{quote}

%--------------------------------------------------------------------------------------
% Abstract
%--------------------------------------------------------------------------------------

\begin{abstract}
We present an improved version of commutator methods for unitary operators under a
weak regularity condition. Once applied to a unitary operator, the method typically
leads to the absence of singularly continuous spectrum and to the local finiteness of
point spectrum. Large families of locally smooth operators are also exhibited. Half
of the paper is dedicated to applications, and a special emphasize is put on the
study of cocycles over irrational rotations. It is apparently the first time that
commutator methods are applied in the context of rotation algebras, for the study of
their generators.
\end{abstract}

\textbf{2000 Mathematics Subject Classification:} 81Q10, 47A35, 47B47.

\smallskip

\textbf{Keywords:} Unitary operators, spectral analysis, Mourre theory, cocycles over
rotations.

%\newpage
%\tableofcontents
%\newpage

%--------------------------------------------------------------------------------------
\section{Introduction}\label{Sec_Intro}
\setcounter{equation}{0}
%--------------------------------------------------------------------------------------

It is commonly accepted that commutator methods (in the sense of E. Mourre) are very
efficient tools for studying spectral properties of self-adjoint operators. These
technics are well developed and have  been formulated with a high degree of
optimality in \cite{ABG}. On the other hand, corresponding approaches for the study
of unitary operators exist but are much less numerous and still at a preliminary
stage of development (see for example \cite{ABCF06,Kat68,Put67,Put81} and references
therein). Accordingly, the aim of this paper is to extend commutator methods for
unitary operators up to an optimality equivalent to the one reached for self-adjoint
operators, and to present some applications.

Up to our knowledge, it is Putnam who presented the first application of commutator
methods for unitary operators. His original result \cite[Thm.~2.3.2]{Put67} reads as
follows\;\!: If $U$ is a unitary operator in a Hilbert space $\H$ and if there exists
a bounded self-adjoint operator $A$ in $\H$ such that $U^*AU-A$ is strictly positive,
then the spectrum of $U$ is purely absolutely continuous. An extension of this
statement to operators $A$ which are only semi-bounded is also presented in that
reference. However, in applications even this assumption of semi-boundedness is often
too restrictive, and a first attempt to extend the theory without this condition has
been proposed in \cite{ABCF06}. To do this, the authors had to impose some regularity
of $U$ with respect to $A$. In the language of Mourre theory for self-adjoint
operators, this regularity assumption corresponds to the inclusion $U\in C^2(A)$.
Apart from the removal of the semi-boundedness assumption, the applicability of the
theory was also significantly broadened by requiring a stronger positivity condition,
but only locally in the spectrum of $U$ and up to a compact term. Under these
assumptions, the spectrum of $U$ has been shown to consist in an absolute continuous
part together with a possible finite set of eigenvalues (see \cite[Thm.~3.3]{ABCF06}
for details).

The main goal of the present work is to weaken the regularity assumption of $U$ with
respect to $A$ to get a theory as optimal as the one in the self-adjoint case.
Accordingly, we show that all the results on the spectrum of $U$ hold if a condition
only slightly stronger than $U\in C^1(A)$ is imposed; namely, either if
$U\in C^{1,1}(A)$ and $U$ has a spectral gap, or if $U\in C^{1+0}(A)$ (see Section
\ref{Sec_Regu} and Theorem \ref{Thm_spec} for notations and details). In addition to
that result, we also obtain a broader class of locally $U$-smooth operators, and thus
more general limiting absorption principles for $U$ (see Proposition \ref{Prop_U}).
Our proofs are rather short and natural, thanks to an extensive use of the Cayley
transform.

The second half of the paper is dedicated to applications. In Sections
\ref{secexample1}-\ref{secexample2}, we consider perturbations of bilateral shifts as
well as perturbations of the free evolution. In both cases, we show that the
regularity assumption can be weaken down to the condition $U\in C^{1+0}(A)$. This
extends significantly some previous results of \cite{ABCF06} in similar situations.
Note that any result on the perturbations of the free evolution is of independent
interest since it provides information on the corresponding Floquet operators (see
\cite{ABCF06,Hua94,HL89} for more details on this issue).

In Section \ref{seccoc}, we determine spectral properties of cocycles over irrational
rotations. It is apparently the first time that commutator methods are applied in
the context of rotation algebras, for the study of their generators. For a class of
cocycles inspired by \cite{ILR93,Med94}, we show that the corresponding unitary
operators have purely Lebesgue spectrum. The result is not new \cite{ILM99} but our
proof is completely new and does not rely on the study of the Fourier coefficients of
the spectral measure. In addition, our approach leads naturally to a limiting
absorption principle and to the obtention of a large class of locally smooth
operators. This information could certainly not be deduced from the study of the
Fourier coefficients alone.

Finally, we treat in Section \ref{secexample4} vector fields on orientable manifolds.
Under suitable assumptions, we show that the unitary operators induced by a general
class of complete vector fields are purely absolutely continuous. This result
complements \cite[Sec.~2.9(ii)]{Put67}, where the author considers the case of
unitary operators induced by divergence-free vector fields on connected open subsets
of $\R^n$.

As a conclusion, let us stress that since unitary operators are bounded, the abstract
theory develop below is slightly simpler than its self-adjoint counterpart. This
feature, together with the fact that various dynamical systems are described by
unitary operators, suggests that the abstract results obtained so far could be applied
in a wide class of situations. For example, it would certainly be of interest to see
if spectral properties of CMV matrices \cite{GKPY11,Sim07}, quantum kicked rotors
\cite{Bell94,Gua09}, quantized Henon maps \cite{FW00,Wei03,Wei04} or integral
transforms \cite{GSL88,Sza01} can be studied with commutator methods.

%--------------------------------------------------------------------------------------
\section{Commutator methods for unitary operators}
\setcounter{equation}{0}
%--------------------------------------------------------------------------------------

%--------------------------------------------------------------------------------------
\subsection{Regularity with respect to $\boldsymbol A$}\label{Sec_Regu}
%--------------------------------------------------------------------------------------

We first recall some facts borrowed from \cite{ABG}. Let $\H$ be a Hilbert space with
norm $\|\;\!\cdot\;\!\|$ and scalar product
$\langle\;\!\cdot\;\!,\;\!\cdot\;\!\rangle$, and denote by $\B(\H)$ the set of
bounded linear operators in $\H$. Now, let $S\in\B(\H)$ and let $A$ be a self-adjoint
operator in $\H$ with domain $\dom(A)$.
For any $k\in\N$, we say that $S$ belongs
to $C^k(A)$, with notation $S\in C^k(A)$, if the map
\begin{equation}\label{re}
\R\ni t\mapsto\e^{-itA}S\e^{itA}\in\B(\H)
\end{equation}
is strongly of class $C^k$. In the case $k=1$, one has $S\in C^1(A)$ if the quadratic
form
$$
\dom(A)\ni\varphi\mapsto\big\langle A\;\!\varphi,S\varphi\big\rangle
-\big\langle\varphi,SA\;\!\varphi\big\rangle\in\C
$$
is continuous for the topology induced by $\H$ on $\dom(A)$. The operator
corresponding to the continuous extension of the form is denoted by $[A,S]\in\B(\H)$
and verifies $[A,S]=\slim_{\tau\to0}[A_\tau,S]$ with
$A_\tau:=(i\tau)^{-1}(\e^{i\tau A}-1)\in\B(\H)$.

Two slightly stronger regularity conditions than $S\in C^1(A)$ are provided by the
following definitions\;\!: $S$ belongs to $C^{1,1}(A)$, with notation
$S\in C^{1,1}(A)$, if
$$
\int_0^1\frac{\d t}{t^2}\;\!
\big\|\e^{-itA}S\e^{itA}+\e^{itA}S\e^{-itA}-2S\big\|<\infty,
$$
and $S$ belongs to $C^{1+0}(A)$, with notation $S\in C^{1+0}(A)$, if $S\in C^1(A)$
and
\begin{equation*}
\int_0^1\frac{\d t}{t}\;\!\big\|\e^{-itA}[A,S]\e^{itA}-[A,S]\big\|<\infty.
\end{equation*}
If we regard $C^1(A)$, $C^{1,1}(A)$, $C^{1+0}(A)$ and $C^2(A)$ as subspaces of $\B(\H)$, then
we have the following inclusions $C^2(A)\subset C^{1+0}(A)\subset C^{1,1}(A)\subset C^1(A)$.

Now, if $H$ is a self-adjoint operator in $\H$ with domain $\dom(H)$, we say that $H$ is of class $C^k(A)$
(respectively $C^{1,1}(A)$ or $C^{1+0}(A)$) if $(H-i)^{-1}\in C^k(A)$ (respectively
$(H-i)^{-1}\in C^{1,1}(A)$ or $(H-i)^{-1}\in C^{1+0}(A)$). If $H$ is of class $C^1(A)$, then the
quadratic form
$$
\dom(A)\ni\varphi\mapsto\big\langle A\;\!\varphi,(H-i)^{-1}\varphi\big\rangle
-\big\langle\varphi,(H-i)^{-1}A\;\!\varphi\big\rangle\in\C
$$
extends continuously to a bounded form defined by the operator
$\big[A,(H-i)^{-1}\big]\in\B(\H)$. Furthermore, the set $\dom(H)\cap\dom(A)$ is a
core for $H$ and the quadratic form
$
\dom(H)\cap\dom(A)\ni\varphi\mapsto\big\langle H\varphi,A\;\!\varphi\big\rangle
-\big\langle A\;\!\varphi,H\varphi\big\rangle
$
is continuous in the topology of $\dom(H)$ \cite[Thm.~6.2.10(b)]{ABG}. This form
extends uniquely to a continuous quadratic form on $\dom(H)$ which can be identified
with a continuous operator $[H,A]$ from $\dom(H)$ to the adjoint space $\dom(H)^*$.
Finally, the following relation holds in $\B(\H):$
\begin{equation}\label{2com}
\big[A,(H-i)^{-1}\big]=(H-i)^{-1}[H,A](H-i)^{-1}.
\end{equation}

%--------------------------------------------------------------------------------------
\subsection{Cayley transform}
%--------------------------------------------------------------------------------------

Let $U$ be a unitary operator in $\H$ with spectrum
$\sigma(U)\subset\T\equiv\big\{\theta\in\C \mid|\theta|=1\big\}$ and (complex) spectral measure
$E^U(\;\!\cdot\;\!)$. Our first result is a restatement of \cite[Thm.~5.1]{ABCF06} in
the framework of the previous subsection. The proof is inspired from the
corresponding proof in the self-adjoint case \cite[Prop.~7.2.10]{ABG}.

\begin{Proposition}[Virial Theorem for $U$]
Let $U$ and $A$ be respectively a unitary and a self-adjoint operator in $\H$, with
$U\in C^1(A)$. Then, $E^U(\{\theta\})\;\!U^*[A,U]\;\!E^U(\{\theta\})=0$ for each
$\theta\in\T$. In particular, one has
$\big\langle\varphi,U^*[A,U]\varphi\big\rangle=0$ for each eigenvector $\varphi$ of
$U$.
\end{Proposition}

\begin{proof}
One has to show that if $\varphi_j\in\H$ satisfies $U\varphi_j=\theta\varphi_j$ for
$j=1,2$, then $\big\langle\varphi_1,U^*[A,U]\varphi_2\big\rangle=0$. But, since
$U^*\varphi_1=\bar\theta\varphi_1$, this follows from the equalities
$$
\big\langle\varphi_1,U^*[A,U]\varphi_2\big\rangle
=\lim_{\tau\to0}\big\langle U\varphi_1,[A_\tau,U]\varphi_2\big\rangle
=\lim_{\tau\to0}\big\{\big\langle\theta\varphi_1,A_\tau\theta\varphi_2\big\rangle
-\big\langle\theta\bar\theta\varphi_1,A_\tau\varphi_2\big\rangle\big\}
=0.
$$
\end{proof}

\begin{Corollary}[Discrete spectrum of $U$]\label{Cor_finite}
Let $U$ and $A$ be respectively a unitary and a self-adjoint operator in $\H$, with
$U\in C^1(A)$. Suppose there exist a Borel set $\Theta\subset\T$, a number $a>0$ and
a compact operator $K\in\K(\H)$ such that
\begin{equation}\label{MourreU}
E^U(\Theta)\;\!U^*[A,U]\;\!E^U(\Theta)\ge aE^U(\Theta)+K.
\end{equation}
Then, the operator $U$ has at most finitely many eigenvalues in $\Theta$, each one of
finite multiplicity.
\end{Corollary}

The following proof is standard but we provide it for completeness\;\!:

\begin{proof}
Let $\varphi\in \H$, $\|\varphi\|=1$, be an eigenvector of $U$ with corresponding
eigenvalue in $\Theta$. Then, it follows from \eqref{MourreU} and from the Virial
theorem that $\big\langle\varphi,K\varphi\big\rangle\le-a$. Now, assume that the
statement of the corollary is false. Then, there exists an infinite orthonormal
sequence $\{\varphi_j\}$ of eigenvectors of $U$ in $E^U(\Theta)\H$. In particular,
one has $\varphi_j\to0$ weakly in $\H$ as $j\to\infty$. And since $K$ is compact,
then $\langle\varphi_j,K\varphi_j\rangle\to0$ as $j\to\infty$, which contradicts the
inequality $\big\langle\varphi_j,K\varphi_j\big\rangle\le-a<0$.
\end{proof}

When an inequality \eqref{MourreU} is satisfied, we say that a Mourre estimate for
$U$ holds on $\Theta$. In the sequel, we suppose that the assumptions of Corollary
\ref{Cor_finite} hold for some open set $\Theta\subset\T$. Therefore, there exists
$\theta\in\Theta$ which is not an eigenvalue of $U$ and the range
$\Ran(1-\bar\theta U)$ of $1-\bar\theta U$ is dense in $\H$. In particular, the
operator
\begin{equation*}
H_\theta\;\!\varphi:=-i\big(1+\bar\theta U\big)\big(1-\bar\theta U\big)^{-1}\varphi,
\qquad\varphi\in\dom(H_\theta):=\Ran(1-\bar\theta U),
\end{equation*}
is self-adjoint due to a standard result on the Cayley transform
\cite[Thm.~8.4(b)]{Wei80}, and $H_\theta$ is bounded if and only if
$\theta\notin\sigma(U)$.

We now prove some simple but useful relations between $U$ and $H_\theta$.

\begin{Lemma}\label{easyrelations}
Suppose that the assumptions of Corollary \ref{Cor_finite} hold for some open set
$\Theta\subset\T$, and let $\theta\in\Theta$ and $H_\theta$ be as above. Then,
\begin{enumerate}
\item[(a)] the operator $H_\theta$ is of class $C^1(A)$ with
$\big[A,(H_\theta-i)^{-1}\big]=-\frac{i\bar \theta}2\;\![A,U]$,
\item[(b)] one has
$
[iH_\theta,A]
=2\;\!\big\{(1-\bar\theta U)^{-1}\big\}^*\;\!U^*[A,U]\;\!(1-\bar\theta U)^{-1}$ in
$\B\big(\dom(H_\theta),\dom(H_\theta)^*\big)$,
\item[(c)] if $U\in C^{1,1}(A)$, then $H_\theta$ is of class $C^{1,1}(A)$,
\item[(d)] if $U\in C^{1+0}(A)$, then $H_\theta$ is of class $C^{1+0}(A)$.
\end{enumerate}
\end{Lemma}

\begin{proof}
Since $(H_\theta-i)^{-1}=\frac{i}{2}(1-\bar\theta U)$, one gets for any
$\varphi\in\dom(A)$
\begin{align*}
\big\langle A\;\!\varphi,(H_\theta-i)^{-1}\varphi\big\rangle
-\big\langle\varphi,(H_\theta-i)^{-1}A\;\!\varphi\big\rangle
&=\textstyle\big\langle A\;\!\varphi,\frac i2(1-\bar\theta U)\varphi\big\rangle
-\big\langle\varphi,\frac i2(1-\bar\theta U)A\;\!\varphi\big\rangle\\
&=\textstyle-\frac{i\bar\theta}2\big\{\big\langle A\;\!\varphi, U\varphi\big\rangle
-\big\langle\varphi,UA\;\!\varphi\big\rangle\big\}\\
&=\textstyle-\frac{i\bar\theta}2\big\langle\varphi,[A,U]\varphi\big\rangle,
\end{align*}
which implies point (a). Then, a successive use of \eqref{2com} and point (a) lead to
the following equalities in $\B\big(\dom(H_\theta),\dom(H_\theta)^*\big)$:
\begin{align*}
i[H_\theta,A]
=i(H_\theta-i)\big[A,(H_\theta-i)^{-1}\big](H_\theta-i)
&=-2\;\!\bar\theta(1-\bar\theta U)^{-1}[A,U](1-\bar\theta U)^{-1}\\
&=2\;\!(1-\theta U^*)^{-1}U^*[A,U](1-\bar\theta U)^{-1}\\
&=2\;\!\big\{(1-\bar\theta U)^{-1}\big\}^*\;\!U^*[A,U]\;\!(1-\bar\theta U)^{-1},
\end{align*}
which imply point (b). Point (c) follows from a direct calculation using the
equality $(H_\theta-i)^{-1}=\frac{i}{2}(1-\bar\theta U)$ which gives
\begin{align*}
&\int_0^1\frac{\d t}{t^2}\;\!
\big\|\e^{-itA}(H_\theta-i)^{-1}\e^{itA}
+\e^{itA}(H_\theta-i)^{-1}\e^{-itA}-2(H_\theta-i)^{-1}\big\|\\
&=\frac12\int_0^1\frac{\d t}{t^2}\;\!
\big\|\e^{-itA}U\e^{itA}+\e^{itA}U\e^{-itA}-2U\big\|.
\end{align*}
The proof of point (d) is similar.
\end{proof}

\begin{Remark}
Due to the simple relation $(H_\theta-i)^{-1}=\frac i2(1-\bar\theta U)$, any
regularity property of $U$ with respect to $A$ is equivalent to the same regularity
property of $(H_\theta-i)^{-1}$ with respect to $A$. In particular, if $U$ belongs to
one regularity class introduced in \cite[Ch.~5]{ABG}, then the resolvent
$(H_\theta-i)^{-1}$ belongs to the same regularity class. This observation might be
useful in applications.
\end{Remark}

Before moving to the next statement, we make two simple observations on the spectrum
$\sigma(H_\theta)$ and the spectral measure $E^{H_\theta}(\;\!\cdot\;\!)$ of
$H_\theta$. First, $\theta'\in\sigma(U)$ if and only if
$-i\;\!\frac{1+\bar\theta\theta'}{1-\bar\theta\theta'}\in\sigma(H_\theta)$ and
$\lambda\in\sigma(H_\theta)$ if and only if
$\theta\;\!\frac{\lambda+i}{\lambda-i}\in\sigma(U)$ (in particular, the point
$\theta'=\theta$ corresponds, as in a stereographic projection with origin $\theta$,
to the points $\lambda=\pm\infty$). Second, for any bounded Borel set $I\subset\R$,
the operator $E^{H_\theta}(I)$ belongs to $\B\big(\H,\dom(H_\theta)\big)$ and extends
by duality to an element of $\B\big(\dom(H_\theta)^*,\H\big)$.

We are now ready to show that a Mourre estimate for $U$ implies a Mourre estimate for
$H_\theta:$

\begin{Proposition}[Mourre estimate for $H_\theta$]\label{Prop_Mourre}
Suppose that the assumptions of Corollary \ref{Cor_finite} hold for some open set
$\Theta\subset\T$, and let $\theta\in\Theta$ and $H_\theta$ be as above. Then, for
any bounded Borel set
$
I\subset\big\{-i\frac{1+\bar\theta\theta'}{1-\bar\theta\theta'}
\mid\theta'\in\Theta\big\}
$
there exists a compact operator $K'\in\K(\H)$ such that
\begin{equation}\label{MourreH}
\textstyle E^{H_\theta}(I)\;\![iH_\theta,A]\;\!E^{H_\theta}(I)
\ge\frac a2\;\!E^{H_\theta}(I)+K',
\end{equation}
with $a>0$ as in \eqref{MourreU}.
\end{Proposition}

\begin{proof}
We know from Lemma \ref{easyrelations}(b) that
$$
E^{H_\theta}(I)\;\![iH_\theta,A]\;\!E^{H_\theta}(I)
=2\;\!E^{H_\theta}(I)\big\{(1-\bar\theta U)^{-1}\big\}^*\;\!U^*[A,U]\;\!
(1-\bar\theta U)^{-1}E^{H_\theta}(I).
$$
Since $(H_\theta-z)^{-1}$ and $U$ commute for each $z\in\C\setminus\R$, one has
$$
(1-\bar\theta U)^{-1}E^{H_\theta}(I)
=E^{H_\theta}(I)(1-\bar\theta U)^{-1}E^{H_\theta}(I).
$$
So, this relation together with its adjoint imply that
$$
E^{H_\theta}(I)\;\![iH_\theta,A]\;\!E^{H_\theta}(I)
=2\;\!E^{H_\theta}(I)\big\{(1-\bar\theta U)^{-1}\big\}^*E^{H_\theta}(I)
\;\!U^*[A,U]\;\!E^{H_\theta}(I)(1-\bar\theta U)^{-1}E^{H_\theta}(I).
$$
Since $E^{H_\theta}(I)=E^U(\Theta)\;\!E^{H_\theta}(I)$, it follows from the Mourre
estimate for $U$ that
$$
E^{H_\theta}(I)\;\![iH_\theta,A]\;\!E^{H_\theta}(I)
\ge2a\;\!E^{H_\theta}(I)\big\{(1-\bar\theta U)^{-1}\big\}^*(1-\bar\theta U)^{-1}
E^{H_\theta}(I)+K'
$$
with $K'\in\K(\H)$. One concludes by noting that
$
\big\{(1-\bar\theta U)^{-1}\big\}^*(1-\bar\theta U)^{-1}
=\frac14|H_\theta-i|^2\ge\frac14
$
holds on $E^{H_\theta}(I)\H$.
\end{proof}

We now prove a limiting absorption principle for the operator $H_\theta$ on the Besov
space $\G:=\big(\dom(A),\H\big)_{1/2,1}$ defined by real interpolation
\cite[Ch.~2]{ABG}. We give two versions of the result: one if we know that $U$ has a
spectral gap and another if we don't know it. We use the notation
$\sigma_{\rm p}(H_\theta)$ for the point spectrum of $H_\theta$.

\begin{Proposition}[Limiting absorption principle for $H_\theta$]\label{Prop_LAP}
Suppose that the assumptions of Corollary \ref{Cor_finite} hold for some open set
$\Theta\subset\T$, and let $\theta\in\Theta$ and $H_\theta$ be as above. Assume also
that
\begin{center}
(i) $U$ has a spectral gap and $U\in C^{1,1}(A)$\qquad or\qquad(ii) $U\in C^{1+0}(A)$.
\end{center}
Then, in any open bounded set
$
I\subset\big\{-i\frac{1+\bar\theta\theta'}{1-\bar\theta\theta'}
\mid\theta'\in\Theta\big\}
$
the operator $H_\theta$ has at most finitely many eigenvalues, each one of finite
multiplicity. Furthermore, for each $\lambda\in I\setminus\sigma_{\rm p}(H_\theta)$
the limits
$
\lim_{\varepsilon\searrow0}(H_\theta-\lambda\mp i\varepsilon)^{-1}
$
exist in the weak* topology of $\B(\G,\G^*)$, uniformly in $\lambda$ on each compact
subset of $I\setminus\sigma_{\rm p}(H_\theta)$. As a corollary, $H_\theta$ has no
singularly continuous spectrum in $I$.
\end{Proposition}

If $U$ has a spectral gap, then the assumption $U\in C^{1,1}(A)$ of (i) is known to
be optimal for $H_\theta$ on the Besov scale $C^{s,p}(A)$ (see the Appendix 7.B of
\cite{ABG}). The assumption $U\in C^{1+0}(A)$ of (ii) is sufficient if $U$ has no
gap, but it is slightly stronger than the $C^{1,1}$-condition. Another approach not
requiring the existence of a gap exists and its regularity assumption is closer to
the $C^{1,1}$-condition than the $C^{1+0}$-condition. However, its implementation is
more involved since it requires the invariance of certain domain under the group
generated by $A$. So, we have decided not to present it for simplicity (see however
\cite[Sec.~7.5]{ABG} for details).

\begin{proof}[Proof of Proposition \ref{Prop_LAP}]
One first observes that $H_\theta$ has at most finitely many eigenvalues in $I$, each
one of finite multiplicity. Indeed, this follows either from the corresponding
statement for $U$ obtained in Corollary \ref{Cor_finite}, or from the Mourre estimate
\eqref{MourreH} and \cite[Cor.~7.2.11]{ABG}. Then, we know from
\cite[Lemma~7.2.12]{ABG} that a strict Mourre estimate holds locally on
$I\setminus\sigma_{\rm p}(H_\theta)$; that is, for any
$\lambda\in I\setminus\sigma_{\rm p}(H_\theta)$ and any
$\delta\in\big(0,\frac a2\big)$, there exists $\varepsilon>0$ such that
$$
\textstyle E^{H_\theta}(\lambda;\varepsilon)\;\![iH_\theta,A]\;\!
E^{H_\theta}(\lambda; \varepsilon)
\ge\big(\frac a2-\delta\big)\;\!E^{H_\theta}(\lambda;\varepsilon),
$$
where
$
E^{H_\theta}(\lambda;\varepsilon)
:= E^{H_\theta}\big((\lambda-\varepsilon,\lambda+\varepsilon)\big)
$.

Once this preliminary observation is made, the limiting absorption principle under
assumption (i) follows from Lemma \ref{easyrelations}(c), Proposition
\ref{Prop_Mourre} and \cite[Thm.~7.4.1]{ABG}, while the same result under assumption
(ii) follows from Lemma \ref{easyrelations}(d), Proposition \ref{Prop_Mourre} and
\cite[Thm.~0.1]{Sah97_2}.

The statement on the singularly continuous spectrum of $H_\theta$ in $I$ is then a
standard consequence of the limiting absorption principle and the finiteness of
$\sigma_{\rm p}(H_\theta)$ in $I$.
\end{proof}

%--------------------------------------------------------------------------------------
\subsection{Absolute continuity of $\;\!U$}
%--------------------------------------------------------------------------------------

We are now in a position to prove our main result on the spectrum of $U:$

\begin{Theorem}[Spectral properties of $U$]\label{Thm_spec}
Let $U$ and $A$ be respectively a unitary and a self-adjoint operator in $\H$. Assume
either that $U$ has a spectral gap and $U\in C^{1,1}(A)$, or that $U\in C^{1+0}(A)$.
Suppose also that there exist an open set $\Theta\subset\T$, a number $a>0$ and a
compact operator $K\in\K(\H)$ such that
\begin{equation}\label{avecK}
E^U(\Theta)\;\!U^*[A,U]\;\!E^U(\Theta)\ge a\;\!E^U(\Theta)+K.
\end{equation}
Then, $U$ has at most finitely many eigenvalues in $\Theta$, each one of finite
multiplicity, and $U$ has no singularly continuous spectrum in $\Theta$.
\end{Theorem}

\begin{proof}
The properties of the eigenvalues of $U$ in $\Theta$ follow directly from Corollary
\ref{Cor_finite}. Now, take $\theta_1\in\Theta\setminus\sigma_{\rm p}(U)$. Then, we
get from Proposition \ref{Prop_LAP} and the correspondence between the spectral
measures of $U$ and $H_{\theta_1}$ (see \cite[Prop.~5.3.10]{BEH08}) that $U$ has no
singularly continuous spectrum in $\Theta\setminus V_1$, where $V_1\subset\Theta$ is
any closed neighborhood of $\theta_1$. The same argument with
$\theta_2\in\Theta\setminus\sigma_{\rm p}(U)$ such that $\theta_2\neq\theta_1$,
implies that $U$ has no singularly continuous spectrum in $\Theta\setminus V_2$,
where $V_2\subset\Theta$ is any closed neighborhood of $\theta_2$. Therefore, if
$V_1$ and $V_2$ are chosen small enough, one has
$\Theta=(\Theta\setminus V_1)\cup(\Theta\setminus V_2)$ and thus $U$ has no
singularly continuous spectrum in $\Theta$.
\end{proof}

\begin{Remark}\label{emptypp}
If the inequality \eqref{avecK} holds with $K=0$, then the operator $U$ has only
purely absolutely continuous spectrum in $\Theta$ (no point spectrum). Indeed, under
this stronger assumption, the inequality \eqref{MourreH} is satisfied with $K'=0$,
and thus a strict Mourre estimate holds for $H_\theta$ locally on $\Theta$. The
statement of Proposition \ref{Prop_LAP} can then be strengthened accordingly, and so
does the statement of Theorem \ref{Thm_spec}.
\end{Remark}

Our next goal is to exhibit locally $U$-smooth operators in our setting. Mimicking
the corresponding definition in the self-adjoint case, we say that an operator
$B\in\B(\H)$ is locally $U$-smooth on an open set $\Theta\subset\T$ if for each
closed set $\Theta'\subset \Theta$
\begin{equation}\label{cond_smooth}
\sum_{n\in\Z}\big\|B\;\!U^nE^U(\Theta')\big\|^2<\infty.
\end{equation}
We also recall from \cite[Thm.~2.2]{ABCF06} that \eqref{cond_smooth} is equivalent to
$$
\sup_{z\in\D,\,\varphi\in\H,\,\|\varphi\|=1}\;\!
\big|\big\langle\varphi,B\;\!\delta(U,z)E^U(\Theta')B^*\varphi\big\rangle\big|
<\infty,
$$
with $\D\subset\C$ the open unit disk and
$
\delta(U,z):=\big(1-zU^*\big)^{-1}-\big(1-{\bar z}^{-1}U^*\big)^{-1}
$
the unitary version of the difference of resolvents. Note that if $\Theta'=\T$ in
\eqref{cond_smooth}, then $B$ is globally $U$-smooth in the usual sense
\cite[Sec.~7]{Kat68}.

Now, let us observe that for $z\in\C$ with $|z|\neq 1$ one has
$$
(1-zU^*)^{-1}
=\left(1-z\bar\theta\;\!\frac{H_\theta-i}{H_\theta+i}\right)^{-1}
=\frac{H_\theta+i}{1-\bar\theta z}
\left(H_\theta+i\;\!\frac{1+\bar\theta z}{1-\bar\theta z}\right)^{-1},
$$
which implies that
\begin{align*}
\delta(U,z)
&=\big(1-|z|^2\big)\big(1-zU^*\big)^{-1}\big\{(1-zU^*)^{-1}\big\}^*\\
&=\big(1-|z|^2\big)\frac{H_\theta+i}{1-\bar \theta z}
\left(\frac{H_\theta+i}{1-\bar \theta z}\right)^*
\left(H_\theta+i\;\!\frac{1+\bar \theta z}{1-\bar \theta z}\right)^{-1}
\bigg\{\left(H_\theta+i\;\!\frac{1+\bar\theta z}{1-\bar\theta z}\right)^{-1}\bigg\}^*\\
&=\big(1-|z|^2\big)\frac{H_\theta^2+1}{|1-\bar \theta z|^2}
\left(H_\theta+i\;\!\frac{1+\bar\theta z}{1-\bar\theta z}\right)^{-1}
\bigg(H_\theta+\overline{i\;\!\frac{1+\bar \theta z}{1-\bar \theta z}}\bigg)^{-1}.
\end{align*}
Moreover, let $\Theta'\subset\T\setminus\{\theta\}$ be closed, and consider the
closed bounded set $J\subset\R$ given by
$$
J:=\left\{-i\;\!\frac{1+\bar \theta\theta'}{1-\bar \theta\theta'}
\mid\theta'\in\Theta'\right\}.
$$
Since $E^U(\Theta')=E^{H_\theta}(J)$, it follows that
$$
H_\theta^2E^U(\Theta')
=H_\theta^2E^{H_\theta}(J)
=\eta(H_\theta)E^{H_\theta}(J)
=\eta(H_\theta)E^{U}(\Theta')
$$
for some function $\eta\in C^\infty(\R;\R)$ with compact support. As a consequence,
one infers that
\begin{equation}\label{Eq_diff_un}
\delta(U,z) E^U(\Theta')
=\big(1-|z|^2\big)\frac{\eta(H_\theta)+1}{|1-\bar \theta z|^2}
\left(H_\theta+i\;\!\frac{1+\bar \theta z}{1-\bar \theta z}\right)^{-1}
\bigg(H_\theta+\overline{i\;\!\frac{1+\bar \theta z}{1-\bar \theta z}}\bigg)^{-1}
E^U(\Theta'),
\end{equation}
where $\eta(H_\theta)\in C^1(A)$ if $H_\theta$ is of class $C^1(A)$, due to
\cite[Thm.~6.2.5]{ABG}. Under the same regularity assumption on $H_\theta$, let us
also observe that for any $s\in[0,1)$ the operator
$$
\langle A\rangle^{-s}\big(\eta(H_\theta)+1\big)\langle A\rangle^s
\quad\hbox{with}\quad\langle A\rangle:=\sqrt{1+A^2},
$$
defined on $\dom\big(\langle A\rangle^s\big)$, extends continuously to an element of
$\B(\H)$ (see \cite[Prop.~5.3.1]{ABG}). With these preparations done, we can prove
the existence of a large class of locally $U$-smooth operators\;\!:

\begin{Proposition}[Locally $U$-smooth operators]\label{Prop_U}
Suppose that the assumptions of Theorem \ref{Thm_spec} hold for some open set
$\Theta\subset\T$. Then, each operator $B\in\B(\H)$ which extends continuously to an
element of $\B\big(\dom(\langle A\rangle^s)^*,\H\big)$ for some $s>1/2$ is locally
$U$-smooth on $\Theta\setminus\sigma_{\rm p}(U)$.
\end{Proposition}

\begin{proof}
We can set $B=\langle A\rangle^{-s}$ and $s\in(1/2,1)$ without lost of generality.
So, we give the proof in this context.

Let $\Theta'\subsetneq\Theta\setminus\sigma_{\rm p}(U)$ be closed and fix
$\theta\in\Theta\setminus\{\sigma_{\rm p}(U)\cup\Theta'\}$. Let also $\O'\subset\D$
be an open set with closure $\cl(\O')$ satisfying
$$
\Theta'\subset\cl(\O')\cap\T
\subset\Theta \setminus\big(\sigma_{\rm p}(U)\cup\{\theta\}\big).
$$
Then, one infers from the definition of $\delta(U,z)$ that
\begin{equation}\label{yeyeyehhh}
\sup_{z\in\D\setminus\O'}\big\|\delta(U,z)E^U(\Theta')\big\|<\infty.
\end{equation}
We can thus restrict our attention to the case $z \in \O'$. But, we know from
\eqref{Eq_diff_un} that
\begin{align*}
&\langle A\rangle^{-s}\delta(U,z)E^U(\Theta')\langle A \rangle^{-s}\nonumber\\
&=\langle A \rangle^{-s} \big(\eta(H_\theta)+1\big)\langle A \rangle^s
\bigg\{\frac{1-|z|^2}{|1-\bar \theta z|^2}\langle A\rangle^{-s}
\left(H_\theta+i\;\!\frac{1+\bar\theta z}{1-\bar\theta z}\right)^{-1}
\bigg(H_\theta+\overline{i\;\!\frac{1+\bar\theta z}{1-\bar\theta z}}\bigg)^{-1}
E^{U}(\Theta')\langle A\rangle^{-s}\bigg\},
\end{align*}
for some $\eta\in C^\infty(\R;\R)$ with compact support. Furthermore, we know that
$\langle A\rangle^{-s}\big(\eta(H_\theta)+1\big)\langle A\rangle^s$ defined
on $\dom\big(\langle A\rangle^s\big)$ extends continuously to an element of $\B(\H)$.
So, we can further restrict our attention to the operator within the curly brackets
for $z\in\O'$.

Using the notations
$
\varepsilon
:=\im\big(-i\;\!\frac{1+\bar\theta z}{1-\bar\theta z}\big)
\equiv\frac{|z|^2-1}{|1-\bar\theta z|^2}
$
and
$
\lambda
:=\re\big(-i\;\!\frac{1+\bar\theta z}{1-\bar\theta z}\big)
\equiv-\frac{2\im(\theta\bar z)}{|1-\bar\theta z|^2}
$,
one gets the following bound for the norm the operator within the curly brackets\;\!:
\begin{equation}\label{youpie}
\big\||\varepsilon|\langle A\rangle^{-s}\big(H_\theta-\lambda-i\varepsilon\big)^{-1}
\big(H_\theta-\lambda+i\varepsilon\big)^{-1}E^{U}(\Theta')\langle A\rangle^{-s}\big\|
\le|\varepsilon|\;\!
\big\|\langle A\rangle^{-s}(H_\theta-\lambda-i\varepsilon)^{-1}\big\|^2.
\end{equation}
Now, the set
$\V:=\big\{-i\frac{1+\bar \theta z}{1-\bar\theta z}\mid z\in\cl(\O')\cap\T\big\}$ is
one of the compact subsets of $\R\setminus\sigma_{\rm p}(H_\theta)$ mentioned in
Proposition \ref{Prop_LAP}. Hence, one infers from Proposition \ref{Prop_LAP} that
\begin{align*}
\sup_{z\in \O'}|\varepsilon|\;\!
\big\|\langle A\rangle^{-s}(H_\theta-\lambda-i\varepsilon)^{-1}\big\|^2
&=\sup_{z\in \O'}\big\|\im\big\{\langle A\rangle^{-s}
(H_\theta-\lambda-i\varepsilon)^{-1}\langle A\rangle^{-s}\big\}\big\|\\
&\le\sup_{z\in \O'}\bigg\|\langle A\rangle^{-s}
\left(H_\theta+i\;\!\frac{1+\bar\theta z}{1-\bar\theta z}\right)^{-1}
\langle A\rangle^{-s}\bigg\|\\
&<\infty.
\end{align*}
This bound, together with \eqref{yeyeyehhh} and \eqref{youpie}, implies the
claim since any closed subset of $\Theta\setminus\sigma_{\rm p}(U)$ is at most the
union of two closed sets as $\Theta'$.
\end{proof}

To close the section, we give a corollary on the perturbations of $U$ which is useful
for applications. For shortness, we state the result in the $C^{1+0}(A)$ case\;\!:

\begin{Corollary}[Perturbations of $U$]\label{Cor_perti}
Let $A$ be a self-adjoint operator in $\H$ and let $U$ and $V$ be unitary operators
in $\H$, with $U,V\in C^{1+0}(A)$. Suppose that there exist an open set
$\Theta\subset\T$, a number $a>0$ and a compact operator $K\in\K(\H)$ such that
\begin{equation}\label{hyp_U}
E^U(\Theta)\;\!U^*[A,U]\;\!E^U(\Theta)\ge a\;\!E^U(\Theta)+K.
\end{equation}
In addition, assume that $V-1$ and $[A,V]$ belong to $\K(\H)$. Then, $VU$ has at most
finitely many eigenvalues in any closed subset of $\Theta$, each one of finite
multiplicity, and $VU$ has no singularly continuous spectrum in $\Theta$.
\end{Corollary}

Note that since $UV=U(VU)\;\!U^*$, the operators $UV$ and $VU$ are unitarily
equivalent and thus have the same spectral properties.

\begin{proof}
The product $VU$ belongs to $C^{1+0}(A)$, since $U$ and $V$ belong to $C^{1+0}(A)$
(see \cite[Prop.~5.2.3(b)]{ABG}).  Furthermore, a direct computation using the
inclusion $[A,V]\in\K(\H)$ implies that
\begin{equation}\label{eq1}
(VU)^*[A,VU]-U^*[A,U]\in\K(\H).
\end{equation}
Now, since $VU-U=(V-1)U$ is compact, it follows from an application of the
Stone-Weierstrass theorem that $\eta(VU)-\eta(U)$ is compact for any $\eta\in C(\T)$.
Therefore, for any $\eta\in C(\T;\R)$ satisfying $\eta(U)=\eta(U)E^U(\Theta)$, one
infers from \eqref{hyp_U} and \eqref{eq1} that
$$
\eta(VU)\;\!(VU)^*[A,VU]\;\!\eta(VU)
=\eta(U)E^U(\Theta)\;\!U^*[A,U]E^U(\Theta)\;\!\eta(U)+K_1
\ge a\;\!\eta^2(U)+K_2
=a\;\!\eta^2(VU)+K_3,
$$
with $K_1,K_2,K_3\in\K(\H)$. Since for each open set $\Theta'\subset\T$ with closure
in $\Theta$, one can find $\eta\in C(\T;\R)$ such that $\eta(U)=\eta(U)E^U(\Theta)$
and $E^{VU}(\Theta')=E^{VU}(\Theta')\eta(VU)$, it follows that
$$
E^{VU}(\Theta')\;\!(VU)^*[A,VU]\;\!E^{VU}(\Theta')
=E^{VU}(\Theta')\;\!\eta(VU)(VU)^*[A,VU]\;\!\eta(VU)E^{VU}(\Theta')
\ge a\;\!E^{VU}(\Theta')+K_4,
$$
with $K_4\in\K(\H)$. Thus, the claims for $VU$ follow directly from Theorem
\ref{Thm_spec}.
\end{proof}

\begin{Example}\label{Ex_compact}
A typical choice for the perturbation $V$ in Corollary \ref{Cor_perti} is $V=\e^{iB}$
with $B=B^*\in\K(\H)$ satisfying $B\in C^{1+0}(A)$ and $[A,B]\in\K(\H)$. Indeed, in
such case we know from \cite[Prop.~5.1.5]{ABG} that $B^k\in C^1(A)$ and that
$\big[A,B^k\big]=\sum_{\ell=0}^{k-1}B^{k-1-\ell}[A,B]B^\ell$ for each $k\in\N^*$.
Therefore, one has for each $\varphi\in\dom(A)$
$$
\big\langle A\;\!\varphi,\e^{iB}\varphi\big\rangle
-\big\langle\varphi,\e^{iB}A\;\!\varphi\big\rangle
=\sum_{k\ge1}\frac{i^k}{k!}\big\langle A\;\!\varphi,B^k\varphi\big\rangle
-\big\langle\varphi,B^kA\;\!\varphi\big\rangle
=\left\langle\varphi,\sum_{k\ge1}\frac{i^k}{k!}
\sum_{\ell=0}^{k-1}B^{k-1-\ell}[A,B]B^\ell\varphi\right\rangle,
$$
where $\sum_{k\ge1}\frac{i^k}{k!}\sum_{\ell=0}^{k-1}B^{k-1-\ell}[A,B]B^\ell$ is a
norm convergent sum of compact operators. It follows that $\e^{iB}\in C^1(A)$ with
$
\big[A,\e^{iB}\big]
=\sum_{k\ge1}\frac{i^k}{k!}\sum_{\ell=0}^{k-1}B^{k-1-\ell}[A,B]B^\ell
\in\K(\H)
$.
Then, a simple calculation using the expression for $\big[A,\e^{iB}\big]$, the
inclusion $B\in C^{1+0}(A)$ and the fact that $C^{1+0}(A)$ is a Banach space, shows
that $\e^{iB}\in C^{1+0}(A)$. Since one also has $\e^{iB}-1\in \K(\H)$, all the
assumptions on $V$ of Corollary \ref{Cor_perti} are satisfied.
\end{Example}

%--------------------------------------------------------------------------------------
\section{Applications}
\setcounter{equation}{0}
%--------------------------------------------------------------------------------------

%--------------------------------------------------------------------------------------
\subsection{Perturbations of bilateral shifts}\label{secexample1}
%--------------------------------------------------------------------------------------

Consider a unitary operator $U$ in a Hilbert space $\H$ and a subspace $\M\subset\H$
such that
$$
\M\perp U^n(\M)~~\hbox{for each }n\in\Z\setminus\{0\}
\qquad\hbox{and}\qquad
\H=\bigoplus_{n\in\Z}U^n(\M).
$$
Such a unitary operator $U$, called a bilateral shift on $\H$ with wandering subspace
$\M$, appears in various instances as in F. and M. Riesz theorem or in ergodic theory
(see \cite[Sec.~3]{Put81}). Using the notation $\varphi\equiv\{\varphi_n\}$ for
elements of $\H$, we define the (number) operator
$$
\textstyle
A\;\!\varphi:=\{n\;\!\varphi_n\},\qquad\varphi\in\dom(A)
:=\left\{\psi\in\H\mid\sum_{n\in\Z}n^2\;\!\|\psi_n\|^2<\infty\right\}.
$$
The operator $A$ is self-adjoint since it can be regarded as a maximal multiplication
operator acting in a $\ell^2$-space. Furthermore, a direct calculation shows that
$
\big\langle A\;\!\varphi,U\varphi\big\rangle
-\big\langle\varphi,UA\;\!\varphi\big\rangle
=\big\langle\varphi,U\varphi\big\rangle
$
for each $\varphi\in\dom(A)$, meaning that $U\in C^2(A)\subset C^{1+0}(A)$ and that
$U^*[A,U]=U^*U=1$.

Therefore, we infer from Theorem \ref{Thm_spec} and Remark \ref{emptypp} that $U$ has
purely absolutely continuous spectrum, as it is well known. Now, let $V$ be another
unitary operator with $V\in C^{1+0}(A)$, $V-1\in\K(\H)$ and $[A,V]\in\K(\H)$. Then,
we deduce from Corollary \ref{Cor_perti} that the operator $VU$ has purely absolutely
continuous spectrum  except at a possible finite number of points in $\T$, where $VU$
has eigenvalues of finite multiplicity.

This result generalizes at the same time the results of \cite[Sec.~3]{Put81} and
\cite[Sec.~4.1]{ABCF06} (in Theorem 4.1 of \cite{ABCF06}, a possible wandering
subspace of the Hilbert space $\ltwo(\R^m)$ is the set of states having support in
$\{x+ty\in\R^m\mid x\cdot y=0,~t\in[0,T]\}$).

%--------------------------------------------------------------------------------------
\subsection{Perturbations of the free evolution}\label{secexample2}
%--------------------------------------------------------------------------------------

In this subsection, we present an extension of Theorem 4.2 of \cite{ABCF06}. So, for
any $m\in\N^*$ let $\H:=\ltwo(\R^m)$ and let $Q:=(Q_1,\ldots,Q_m)$ and
$P:=(P_1,\ldots,P_m)$ be the usual families of momentum and position operators
defined
by
$$
(Q_j\varphi)(x):=x_j\varphi(x)
\qquad\hbox{and}\qquad
(P_j\varphi)(x):=-i(\partial_j\varphi)(x)
$$
for any Schwartz function $\varphi\in\S(\R^m)$ and $x\in\R^m$. Our aim is to treat
perturbations of the free evolution $U:=\e^{-iTP^2}$, which is known to have purely
absolutely continuous spectrum $\sigma(U)=\sigma_{\rm ac}(U)=\T$ for each $T>0$.
Following \cite[Def.~1.1]{Yok98}, we consider the operator
$$
A\;\!\varphi
:=\frac12\;\!\big\{(P^2+1)^{-1}P\cdot Q+Q\cdot P(P^2+1)^{-1}\big\}\;\!\varphi,
\qquad\varphi\in\S(\R^m),
$$
which is essentially self-adjoint due to \cite[Prop.~7.6.3(a)]{ABG}. Then, a direct
calculation shows that
$
\big\langle A\;\!\varphi,U\varphi\big\rangle
-\big\langle\varphi,UA\;\!\varphi\big\rangle
=\big\langle\varphi,2TUP^2(P^2+1)^{-1}\varphi\big\rangle
$
for each $\varphi\in\S(\R^m)$, meaning that $[A,U]=2TUP^2(P^2+1)^{-1}$ by the density
of $\S(\R^m)$ in $\dom(A)$. Therefore, one has $U\in C^1(A)$ with
$U^*[A,U]=2TP^2(P^2+1)^{-1}$, and further computations on $\S(\R^m)$ show that
$U\in C^2(A)\subset C^{1+0}(A)$. Now, observe that for any open set $\Theta\subset\T$
with closure $\cl(\Theta)$ satisfying $\cl(\Theta)\cap\{1\}=\varnothing$, there
exists $\delta>0$ such that
$$
E^U(\Theta)
=E^{P^2}\big([\delta,\infty)\big)E^U(\Theta)
$$
and
$$
E^U(\Theta)\;\!U^*[A,U]\;\!E^U(\Theta)
=2TE^U(\Theta)P^2(P^2+1)^{-1}E^{P^2}\big([\delta,\infty)\big)E^U(\Theta)
\ge2T\delta(\delta+1)^{-1}E^U(\Theta).
$$
By combining what precedes together with Corollary \ref{Cor_perti}, we obtain the
following\;\!:

\begin{Lemma}
Let $V$ be a unitary operator in $\H$ satisfying $V\in C^{1+0}(A)$, $V-1\in\K(\H)$
and $[A,V]\in\K(\H)$. Then, the eigenvalues of the operator $VU$ outside $\{1\}$ are
of finite multiplicity and can accumulate only at $\{1\}$. Furthermore, $VU$ has no
singularly continuous spectrum.
\end{Lemma}

%--------------------------------------------------------------------------------------
\subsection{Cocycles over irrational rotations}\label{seccoc}
%--------------------------------------------------------------------------------------

Consider two unitary operators $U$ and $V$ in a Hilbert space $\H$ satisfying, for
some irrational $\theta\in[0,1)$, the commutation relation
\begin{equation}\label{eq_rot}
UV=\e^{2\pi i\theta}VU.
\end{equation}
The universal $C^*$-algebra $A_\theta$ generated by such a pair of unitaries is known
as the irrational rotation algebra \cite[Sec.~12.3]{W-O93} and has attracted a lot
attention these last years. It is known that both $U$ and $V$ have full spectrum,
that is,
$$
\sigma(U)=\sigma(V)=\T.
$$

Our aim is to study the nature of the spectrum of $U$ when the pair $(U,V)$ is the
following\;\!: Let $\H:=\ltwo\big([0,1)\big)$ be the $\ltwo$-space over the interval
$[0,1)$ with addition modulo $1$, and let $V\in\B(\H)$ be the unitary operator of
multiplication by the independent variable, \ie
$$
(V\varphi)(x):=\e^{2\pi ix}\varphi(x)
\quad\hbox{for all }\varphi\in\H\hbox{ and a.e. }x \in [0,1).
$$
Given $f:[0,1)\to\R$ a measurable function, let also $U\in\B(\H)$ be the unitary
operator associated with the cocycle $f$ over the rotation by $\theta$, \ie
\begin{equation}\label{defdeV}
(U\varphi)(x):=\e^{2\pi if(x)}\varphi([x+\theta])
\quad\hbox{for all }\varphi\in\H\hbox{ and a.e. }x \in [0,1),
\end{equation}
where $[y]:=y$ modulo $1$. Then, $U$ and $V$ satisfy the commutation relation
\eqref{eq_rot}, and the spectrum of $U$ is either purely punctual, purely singularly
continuous or purely Lebesgue \cite[Thm.~3]{Hel86} (the precise nature of the
spectrum highly depends on properties of the pair $(f,\theta)$, see for example
\cite{ILM99,ILR93,Med94}).

In the sequel, we treat the case $f:=m\id+h$, with $m\in\Z^*$, $\id:[0,1)\to[0,1)$
the identity function and $h:[0,1)\to\R$ an absolutely continuous function satisfying
$h(0)=h(1)$. For this, we introduce the self-adjoint operator $P$ in $\H$ defined by
$$
P\varphi:=-i\varphi',\quad
\varphi\in\dom(P):=\big\{\varphi \in \H\mid \varphi \hbox{ is absolutely continuous},
~\varphi'\in\ltwo\big([0,1)\big)~\hbox{and}~\varphi(0)=\varphi(1)\big\}.
$$
Then, we observe that an integration by parts gives for each $\varphi \in \dom(P)$
that
\begin{equation*}
\big\langle P\varphi,U\varphi\big\rangle-\big\langle\varphi,UP\varphi\big\rangle
=\big\langle\varphi,2\pi\big(m+h'\big)U\varphi\big\rangle,
\end{equation*}
with $h'$ the operator of multiplication by $h'$. Therefore, if
$h'\in\linf\big([0,1)\big)$, one infers that $U\in C^1(P)$ with
$U^*[P,U]=2\pi\big(m+h'([\,\cdot\;\!-\theta])\big)$. Moreover, by imposing some
additional condition on the size of $h'$, one would also obtain that $U^*[P,U]$ has a
definite sign (the one of $m$). However, since this is not very satisfactory, we
shall explain in the next proposition how this additional condition can be avoided by
modifying adequately the conjugate operator $P$. The trick is based on the ergodic
theorem and the new conjugate operator is of the form
\begin{equation}\label{PV}
P_n := \frac1n\sum_{j=0}^{n-1} U^{-j}PU^j
\end{equation}
for some large $n\in\N^*$ (see the Appendix for abstract results on the
self-adjointness of $P_n$ and on regularity properties with respect to $P_n$).

\begin{Proposition}\label{prop_cocycles}
Let $f:=m\id+h$, with $m\in\Z^*$ and $h\in C^1\big([0,1);\R\big)$ with $h'$
Dini-continuous and $h(0)=h(1)$. Then the operator $U$, defined by \eqref{defdeV}
with $\theta$ irrational, has a purely Lebesgue spectrum equal to $\T$. Furthermore,
the operator $\langle P\rangle^{-s}$ is globally $U$-smooth for any $s>1/2$.
\end{Proposition}

\begin{proof}
In what follows, we assume without loss of generality that $m\in\N^*$ (if $m<0$, the
same proof works with $P$ replaced by $-P$).

Since $h'$ is Dini-continuous, $h'\in\linf\big([0,1)\big)$ and $U\in C^1(P)$ as
indicated before. Thus, one infers from the Appendix (with $A=P$) that for each
$n\in\N^*$ the operator $P_n$, defined in \eqref{PV}, is self-adjoint on
$\dom(P_n):=\dom(P)$. Moreover, one deduces from Lemma \ref{lemPn}(a) that $U$
belongs to $C^1(P_n)$ with
\begin{equation*}
[P_n,U]
=\frac1n\sum_{j=0}^{n-1}U^{-j}[P,U]U^j
=\frac{2\pi}n\sum_{j=0}^{n-1}U^{-j}\big(m +h'\big)U^{j+1}
=2\pi U\Bigg(m+\frac1n\sum_{j=1}^nh'([\,\cdot\;\!-j\theta])\Bigg).
\end{equation*}
Thus, one has
$U^*[P_n,U]=2\pi\big(m+\frac1n\sum_{j=1}^nh'([\,\cdot\;\!-j\theta])\big)$. Also,
since $\int_0^1\d x\,h'(x)=0$, it follows from the strict ergodicity of the
irrational translation by $\theta$ that
$\big|\frac1n\sum_{j=1}^nh'([\,\cdot\;\!-j\theta])\big|<1/2$ for $n$ large enough.
So, one deduces that $U^*\big[P_n,U\big]\geq\pi$ if $n$ is large enough.

To check that $U\in C^{1+0}(P_n)$, we observe from Lemma \ref{lemPn}(b) that it is
sufficient to show the inclusion $U \in C^{1+0}(P)$. We also recall that
$[P,U]=2\pi\big(m+h'\big)U$. Therefore, since $U$ satisfies the condition
$\int_0^1\frac{\d t}t\;\!\big\|\e^{-itP}U\e^{itP}-U\big\|<\infty$, one is reduced to
showing that
$$
\int_0^1\frac{\d t}{t}\;\!\big\|\e^{-itP}(m+h')\e^{itP}-(m+h')\big\|
=\int_0^1\frac{\d t}{t}\;\!\big\|\e^{-itP}h'\e^{itP}-h'\big\|<\infty.
$$
However, this condition is readily verified since it corresponds to nothing else but
the Dini-continuity of $h'$. One then concludes by applying Theorem \ref{Thm_spec}
and Proposition \ref{Prop_U}, and by taking into account the equality
$\dom(P_n)=\dom(P)$.
\end{proof}

The first part of Proposition \ref{prop_cocycles} is not new; the nature of the
spectrum of $U$ was already determined in \cite{ILM99,ILR93} under  a slightly weaker
assumption ($h$ absolutely continuous with $h'$ of bounded variation). On the other
hand, we have not been able to find in the literature any information about globally
$U$-smooth operators. Therefore, the second part Proposition \ref{prop_cocycles} is
apparently new.

%--------------------------------------------------------------------------------------
\subsection{Vector fields on manifolds}\label{secexample4}
%--------------------------------------------------------------------------------------

Let $M$ be a smooth orientable manifold of dimension $n\ge1$ with volume form
$\Omega$, and let $\H:=\ltwo(M,\Omega)$ be the corresponding $\ltwo$-space. Consider
a complete $C^\infty$ vector field $X\in\X(M)$ on $M$ with flow
$\R\times M\ni(t,p)\mapsto F_t(p)\in M$. Then, there exists for each $t\in\R$ a
unique function $\det_\Omega(F_t)\in C^\infty(M;\R)$ satisfying
$F_t^*\Omega=\det_\Omega(F_t)\Omega$, with $F_t^*$ the pullback by $F_t$
($\det_\Omega(F_t)$ is the determinant of $F_t$, see \cite[Def.~2.5.18]{AM78}).
Furthermore, the continuity of the map $t\mapsto F_t(\;\!\cdot\;\!)$, together with
the fact that $F_0=\id_M$ is orientation-preserving, implies that
$\det_\Omega(F_t)>0$ on $M$ for all $t\in\R$. So, one can define for each $t\in\R$
the operator
$$
U_t\;\!\varphi:=\big\{\det_\Omega(F_t)\big\}^{1/2}F_t^*\varphi,\qquad\varphi\in\H,
$$
which is easily shown to be unitary in $\H$ (use for example
\cite[Prop.~2.5.20]{AM78}). Now, let $M_0:=\big\{p\in M\mid X_p=0\big\}$ be the set
of critical points of $X$. Then, the subspaces $\ltwo(M_0,\Omega)$ and
$\ltwo(M\setminus M_0,\Omega)$ reduce $U_t$, with
$U_t\upharpoonright\ltwo(M_0,\Omega)=1$ and
$U_t\upharpoonright\ltwo(M\setminus M_0,\Omega)$ unitary. Therefore, we can restrict
our attention to the unitary operator
$U_t^0:=U_t\upharpoonright\ltwo(M\setminus M_0,\Omega)$ in
$\H_0:=\ltwo(M\setminus M_0,\Omega)$.

Suppose that there exist a function $g\in C^1(M\setminus M_0;\R)$ and a constant
$\delta>0$ such that
\begin{equation}\label{cond_g}
\d g\cdot X\ge\delta
\qquad\hbox{and}\qquad
\d g\cdot X\in\linf(M\setminus M_0,\Omega).
\end{equation}
Then, the multiplication operator
$$
A\;\!\varphi:=-g\;\!\varphi,\qquad\varphi\in C_{\rm c}(M\setminus M_0),
$$
is essentially self-adjoint in $\H_0$ \cite[Ex.~5.1.15]{Ped89}, and a direct
calculation using the inclusion
$U_t^0\;\!C_{\rm c}(M\setminus M_0)\subset C_{\rm c}(M\setminus M_0)$ implies for
each $\varphi\in C_{\rm c}(M\setminus M_0)$ that
\begin{align*}
\big\langle A\;\!\varphi,U_t^0\varphi\big\rangle_{\H_0}
-\big\langle\varphi,U_t^0A\;\!\varphi\big\rangle_{\H_0}
=\big\langle\varphi,U_t^0\big(U_{-t}AU_t^0-A\big)\varphi\big\rangle_{\H_0}
&=-\big\langle\varphi,U_t^0\big(F_{-t}^*\;\!g-g\big)\varphi\big\rangle_{\H_0}\\
&=\left\langle\varphi,
U_t^0\int_0^t\d s\,(\d g)_{F_{-s}}\cdot X_{F_{-s}}\;\!\varphi\right\rangle_{\H_0},
\end{align*}
where $\int_0^t\d s\,(\d g)_{F_{-s}}\cdot X_{F_{-s}}$ is the multiplication operator
by the function $p\mapsto\int_0^t\d s\,(\d g)_{F_{-s}(p)}\cdot X_{F_{-s}(p)}$. This,
together with the density of $C_{\rm c}(M\setminus M_0)$ in $\dom(A)$ and the second
condition of \eqref{cond_g}, implies that
$[A,U_t^0]=U_t^0\int_0^t\d s\,(\d g)_{F_{-s}}\cdot X_{F_{-s}}\in\B(\H_0)$ and
$
\big[A,[A,U_t^0]\big]
=\big(\int_0^t\d s\,(\d g)_{F_{-s}}\cdot X_{F_{-s}}\big)^2\in\B(\H_0)
$.
So, $U_t^0\in C^2(A)\subset C^{1+0}(A)$, and the first condition of \eqref{cond_g}
gives for any $t>0$
$$
\big(U_t^0\big)^*\big[A,U_t^0\big]
=\int_0^t\d s\,(\d g)_{F_{-s}}\cdot X_{F_{-s}}
\ge\delta t.
$$
Therefore, we infer from Theorem \ref{Thm_spec} and Remark \ref{emptypp} that $U_t^0$
is purely absolutely continuous in $\H_0$ for any $t>0$ (and thus for all $t\neq0$
since $U_{-t}^0=(U_t^0)^*$). In particular, $U_t$ is purely absolutely continuous in
$\H$ for any $t\neq0$ if the measure of $M_0$ (relative to $\Omega$) is zero. This
result complements \cite[Sec.~2.9(ii)]{Put67}, where the author treats the case of
unitary operators induced by divergence-free vector fields on connected open subsets
of $\R^n$.

To conclude this subsection, we exhibit an explicit class of vector fields on $\R^n$
satisfying all of our assumptions\;\!: Take $M=\R^n$ and let $X\in\X(\R^n)$ be given
by
$$
X_x:=f(|x|)\;\!\frac x{|x|}\;\!,\quad x\in\R^n,
$$
where $f\in C^\infty\big([0,\infty);\R\big)$ satisfies
$\lim_{r\searrow0}|f(r)/r|<\infty$ and $\delta_1r\le f(r)\le\delta_2r+\delta_3$ for
all $r\ge0$ and some $\delta_1,\delta_2,\delta_3>0$. Then, we have $M_0=\{0\}$, and
we know from \cite[Rem.~1.8.7]{Nar73} that $X$ is complete. Therefore, the
restriction $U_t^0:=U_t\upharpoonright\ltwo(\R^n\setminus\{0\},\d x)$ is a well
defined unitary operator. Now, let $g\in C^1(\R^n\setminus\{0\};\R)$ be given by
$g(x):=\ln(|x|^2)$. Then, the conditions \eqref{cond_g} are verified due to the
properties of the function $f$, and so $U_t^0$ is purely absolutely continuous in
$\ltwo(\R^n\setminus\{0\},\d x)$ for any $t\neq0$. Since $M_0$ has Lebesgue measure
zero, it follows that $U_t$ is purely absolutely continuous in $\ltwo(\R^n,\d x)$.
This covers for instance the well known case of the dilation group, where $f(r)=r$.

%--------------------------------------------------------------------------------------
\section{Appendix}
%--------------------------------------------------------------------------------------

In this appendix, we introduce an abstract class of conjugate operators which is
useful for the study of cocycles in Section \ref{seccoc}.

Let $A$ and $U$ be respectively a self-adjoint and a unitary operator in a Hilbert
space $\H$. If $U\in C^1(A)$, then we know from \cite[Prop.~5.1.5-5.1.6]{ABG} that
$U^k\in C^1(A)$ for each $k\in\Z$ and $U^k\dom(A)=\dom(A)$. Therefore, for each
$n\in\N^*$ the operator $\frac1n\sum_{j=0}^{n-1}U^{-j}\big[A,U^j\big]$ is bounded,
and the operator
$$
A_n\varphi
:=\frac1n\sum_{j=0}^{n-1}U^{-j}AU^j\varphi
\equiv\frac1n\sum_{j=0}^{n-1}U^{-j}\big[A,U^j\big]\varphi+A\varphi,
\qquad\varphi\in\dom(A_n):=\dom(A),
$$
is self-adjoint. Furthermore, we have the following result on the regularity of $U$
with respect to $A_n$:

\begin{Lemma}\label{lemPn}
Take $n\in\N^*$ and let $A$ and $U$ be respectively a self-adjoint and a unitary
operator in a Hilbert space $\H$. Then,
\begin{enumerate}
\item[(a)] if $\;\!U\in C^1(A)$, then $U\in C^1(A_n)$ with
$[A_n,U]=\frac1n\sum_{j=0}^{n-1}U^{-j}[A,U]U^j$,
\item[(b)] if $\;\!U\in C^{1+0}(A)$, then $U\in C^{1+0}(A_n)$.
\end{enumerate}
\end{Lemma}

\begin{proof}
(a) Since $U^k\in C^1(A)$ for each $k\in\Z$, one has for any $\varphi\in\dom(A_n)$ that
$$
\big\langle A_n\varphi,U\varphi\big\rangle-\big\langle\varphi,UA_n\varphi\big\rangle
=\frac1n\sum_{j=0}^{n-1}
\big\langle\varphi,\big(U^{-j}AU^{j+1}-U^{1-j}AU^j\big)\varphi\big\rangle
=\frac1n\sum_{j=0}^{n-1}\big\langle\varphi,U^{-j}[A,U]U^j\varphi\big\rangle,
$$
with $[A,U]\in\B(\H)$. This implies the claim.

(b) We know from point (a) that $U\in C^1(A_n)$ with
$[A_n,U]=\frac1n\sum_{j=0}^{n-1}U^{-j}[A,U]U^j$. Therefore, it is sufficient to show
for each $j\in\{1,\ldots,n-1\}$ that
$$
\int_0^1\frac{\d t}{t}\;\!\big\|\e^{-itA_n}B_j\e^{itA_n}-B_j\big\|<\infty,
$$
with $B_j:=U^{-j}[A,U]U^j$. But, for each $t\in\R$ and each $\varphi\in\dom(A)$ we
have
$$
\e^{itA_n}\varphi-\e^{itA}\varphi
=\int_0^t\d s\,\frac\d{\d s}\big(\e^{isA_n}\e^{-isA}-1\big)\e^{itA}\varphi
=\frac in\sum_{k=0}^{n-1}
\int_0^t\d s\,\e^{isA_n}U^{-k}\big[A,U^k\big]\e^{i(t-s)A}\varphi.
$$
So, there exists $C_t\in\B(\H)$ with $\|C_t\|\le{\rm Const.}|t|$ such that
$\e^{itA_n}=\e^{itA}+C_t$, and thus
\begin{equation}\label{2_integrales}
\int_0^1\frac{\d t}{t}\;\!\big\|\e^{-itA_n}B_j\e^{itA_n}-B_j\big\|
\le{\rm Const}+\int_0^1\frac{\d t}{t}\;\!\big\|\e^{-itA}B_j\e^{itA}-B_j\big\|.
\end{equation}
Now, the integral $\int_0^1\frac{\d t}{t}\;\!\big\|\e^{-itA}D\e^{itA}-D\big\|$ is
finite for $D=U^{-j}$, $D=[A,U]$ and $D=U^j$ due to the assumption. So, the integral
in the r.h.s. of \eqref{2_integrales} is also finite, and thus the claim is proved.
\end{proof}

%--------------------------------------------------------------------------------------
\section*{Acknowledgements}
%--------------------------------------------------------------------------------------

C.F. and R.T.d.A. thank Olivier Bourget for many useful discussions on the spectral
theory for unitary operators. S.R. is grateful for the hospitality provided by the
Mathematics Department of the Pontificia Universidad Cat\'olica de Chile in November
2011.

%--------------------------------------------------------------------------------------
%\bibliography{../bibliographie/bibliographie}

\begin{thebibliography}{1}

\bibitem{AM78}
R.~Abraham and J.~E. Marsden.
\newblock {\em Foundations of mechanics}. Second edition,
\newblock Benjamin/Cummings Publishing Co. Inc. Advanced Book Program, Reading,
  Mass., 1978.

\bibitem{Amr09}
W.~O. Amrein.
\newblock {\em Hilbert space methods in quantum mechanics}.
\newblock Fundamental Sciences. EPFL Press, Lausanne, 2009.

\bibitem{ABG}
W.~O. Amrein, A.~{{Boutet de Monvel}}, and V.~Georgescu.
\newblock {\em ${C_0}$-groups, commutator methods and spectral theory of
  ${N}$-body Hamiltonians}. Volume 135 of {\em Progress in Math.}
\newblock Birkh\"auser, Basel, 1996.

\bibitem{ABCF06}
M.~A. Astaburuaga, O.~Bourget, V.~H. Cort{\'e}s, and C.~Fern{\'a}ndez.
\newblock Floquet operators without singular continuous spectrum.
\newblock {\em J. Funct. Anal.} 238(2): 489--517, 2006.

\bibitem{Bell94}
J.~Bellissard.
\newblock Noncommutative methods in semiclassical analysis.
\newblock In {\em Transition to chaos in classical and quantum mechanics
  ({M}ontecatini {T}erme, 1991)}, volume 1589 of {\em Lecture Notes in Math.},
  pages 1--64. Springer, Berlin, 1994.

\bibitem{BEH08}
J.~Blank, P.~Exner, and M.~Havl{\'{\i}}{\v{c}}ek.
\newblock {\em Hilbert space operators in quantum physics}.
\newblock Theoretical and Mathematical Physics. Springer, New York, second
  edition, 2008.

\bibitem{FW00}
J.~E. Forn{\ae}ss and B.~Weickert.
\newblock A quantized {H}\'enon map.
\newblock {\em Discrete Contin. Dynam. Systems} 6(3): 723--740, 2000.

\bibitem{GSL88}
D.~S. Gilliam, J.~R. Schulenberger, and J.~L. Lund.
\newblock Spectral representation of the {L}aplace and {S}tieltjes transforms.
\newblock {\em Mat. Apl. Comput.} 7(2): 101--107, 1988.


\bibitem{GKPY11}
L.~Golinskii, A.~Kheifets, F.~Peherstorfer, and P.~Yuditskii.
\newblock Scattering theory for {CMV} matrices: uniqueness, {H}elson-{S}zeg\"o
  and strong {S}zeg\"o theorems.
\newblock {\em Integral Equations Operator Theory} 69(4): 479--508, 2011.

\bibitem{Gua09}
I.~Guarneri.
\newblock On the spectrum of the resonant quantum kicked rotor.
\newblock {\em Ann. Henri Poincar\'e} 10(6): 1097--1110, 2009.

\bibitem{Hel86}
H.~Helson.
\newblock Cocycles on the circle.
\newblock {\em J. Operator Theory} 16(1): 189--199, 1986.

\bibitem{Hua94}
M.~J. Huang.
\newblock On the absolutely continuous subspaces of {F}loquet operators.
\newblock {\em Proc. Roy. Soc. Edinburgh Sect. A} 124(4): 703--712, 1994.

\bibitem{HL89}
M.~J. Huang and R.~B. Lavine.
\newblock Boundedness of kinetic energy for time-dependent {H}amiltonians.
\newblock {\em Indiana Univ. Math. J.} 38(1): 189--210, 1989.

\bibitem{ILM99}
A.~Iwanik, M.~Lema{\'n}czyk, and C.~Mauduit.
\newblock Piecewise absolutely continuous cocycles over irrational rotations.
\newblock {\em J. London Math. Soc. (2)}, 59(1): 171--187, 1999.

\bibitem{ILR93}
A.~Iwanik, M.~Lema\'nczyk, and D.~Rudolph.
\newblock Absolutely continuous cocycles over irrational rotations.
\newblock {\em Israel J. Math.} 83(1-2): 73--95, 1993.

\bibitem{Kat68}
T.~Kato.
\newblock Smooth operators and commutators.
\newblock {\em Studia Math.} 31: 535--546, 1968.

\bibitem{Med94}
H.~A. Medina.
\newblock Spectral types of unitary operators arising from irrational rotations
  on the circle group.
\newblock {\em Michigan Math. J.} 41(1): 39--49, 1994.

\bibitem{Nar73}
R.~Narasimhan.
\newblock {\em Analysis on real and complex manifolds}. Second edition,
\newblock Masson\thinspace \&\thinspace Cie, \'Editeur, Paris, 1973.

\bibitem{Ped89}
G.~K. Pedersen.
\newblock {\em Analysis now}. Volume 118 of {\em Graduate Texts in
  Mathematics}.
\newblock Springer-Verlag, New York, 1989.

\bibitem{Put67}
C.~R. Putnam.
\newblock {\em Commutation properties of {H}ilbert space operators and related
  topics}.
\newblock Ergebnisse der Mathematik und ihrer Grenzgebiete, Band 36.
  Springer-Verlag New York, Inc., New York, 1967.

\bibitem{Put81}
C.~R. Putnam.
\newblock Absolute continuity and hyponormal operators.
\newblock {\em Internat. J. Math. Math. Sci.} 4(2): 321--335, 1981.

\bibitem{Sah97_2}
J.~Sahbani.
\newblock The conjugate operator method for locally regular {H}amiltonians.
\newblock {\em J. Operator Theory} 38(2): 297--322, 1997.

\bibitem{Sim07}
B.~Simon.
\newblock C{MV} matrices: five years after.
\newblock {\em J. Comput. Appl. Math.} 208(1): 120--154, 2007.

\bibitem{Sza01}
D.~C. Szajda.
\newblock Absolute continuity of a class of unitary operators.
\newblock {\em Houston J. Math.} 27(1): 189--202, 2001.

\bibitem{W-O93}
N.~E. Wegge-Olsen.
\newblock {\em {$K$}-theory and {$C^*$}-algebras, A friendly approach}.
\newblock Oxford Science Publications. The Clarendon Press Oxford University Press, New York, 1993.

\bibitem{Wei03}
B.~Weickert.
\newblock Quantizations of linear self-maps of {$\mathbb R^2$}.
\newblock {\em Acta Sci. Math. (Szeged)} 69(3-4): 619--631, 2003.

\bibitem{Wei04}
B.~Weickert.
\newblock Spectral properties and dynamics of quantized {H}enon maps.
\newblock {\em Trans. Amer. Math. Soc.} 356(12): 4951--4968, 2004.

\bibitem{Wei80}
J.~Weidmann.
\newblock {\em Linear operators in {H}ilbert spaces}. Volume~68 of {\em
  Graduate Texts in Mathematics}.
\newblock Springer-Verlag, New York, 1980.

\bibitem{Yok98}
K.~Yokoyama.
\newblock Mourre theory for time-periodic systems.
\newblock {\em Nagoya Math. J.} 149: 193--210, 1998.

\end{thebibliography}
%--------------------------------------------------------------------------------------

\def\cprime{$'$} \def\cprime{$'$}

%--------------------------------------------------------------------------------------

\end{document}